\documentclass[11pt]{article}
\usepackage{fullpage} 
\usepackage{url}

\usepackage{amssymb,comment}
\usepackage{amsmath}
\usepackage{tikz}
\usepackage{pgflibrarysnakes}
\usetikzlibrary{snakes}
\usepackage{times}
\usepackage{amstext}
\usepackage{calc}
\usepackage{amsopn}
\usepackage[noend]{algorithmic}
\usepackage[boxed]{algorithm}
\usepackage{eucal}
\usepackage{latexsym}
\usepackage[rflt]{floatflt}
\usepackage{wrapfig}
\usepackage{tabularx}

\usepackage{amsthm}

\newtheorem{theorem}{Theorem}
\newtheorem{corollary}[theorem]{Corollary}

\newtheorem{lemma}[theorem]{Lemma}
\newtheorem{conjecture}[theorem]{Conjecture}

\theoremstyle{definition}
\newtheorem{definition}[theorem]{Definition}

{\bf}{\rm}

\newcommand{\stubname}{{\sc{Stubborn problem}}}
\newcommand{\ccname}{{\sc{Compatible colouring}}}
\newcommand{\listpartitionname}{{\sc{List Matrix Partition}}}

\newcommand{\kccname}{{k-\sc{Compatible colouring}}}
\newcommand{\cspname}{{\sc{Constraint Satisfaction Problem}}}
\newcommand{\threeccname}{{3-\sc{Compatible colouring}}}
\newcommand{\twoccname}{{2-\sc{Compatible colouring}}}
\newcommand{\cspshort}{{\sc{CSP}}}
\newcommand{\cspfull}{full-{\sc{CSP}}}
\newcommand{\cspgamma}{{\sc{CSP($\Gamma$)}}}

\floatname{algorithm}{Algorithm}

\newcommand{\defproblemu}[3]{
  \vspace{1mm}
\noindent\fbox{
  \begin{minipage}{\textwidth}
  #1 \\ 
  {\bf{Input:}} #2  \\
  {\bf{Question:}} #3
  \end{minipage}
  }
  \vspace{1mm}
}

\newcommand{\ra}{{\rightarrow}}

\newcommand{\RR}{ {\mathcal{R} }}
\newcommand{\GG}{ {\mathcal{G} }}
\newcommand{\BB}{ {\mathcal{B} }}
\newcommand{\cC}{ {\mathcal{X}}}
\newcommand{\cX}{ {\mathcal{Y}}}
\newcommand{\cY}{ {\mathcal{Z}}}

\newcommand{\TS} { {\texttt{ToSet}}}
\newcommand{\BTS}{ {\texttt{ToSet}^{\BB}}}
\newcommand{\RTS}{ {\texttt{ToSet}^{\RR}}}
\newcommand{\GTS}{ {\texttt{ToSet}^{\GG}}}
\newcommand{\CTS}{ {\texttt{ToSet}^{\cC}}}
\newcommand{\XTS}{ {\texttt{ToSet}^{\cX}}}
\newcommand{\YTS}{ {\texttt{ToSet}^{\cY}}}

\newcommand{\RTD}{ {\texttt{ToDo}^{\RR}}}

\newcommand{\CTD}{ {\texttt{ToDo}^{\cC}}}
\newcommand{\XTD}{ {\texttt{ToDo}^{\cX}}}

\newcommand{\CS}{ {\texttt{Set}^{\cC}}}
\newcommand{\XS}{ {\texttt{Set}^{\cX}}}
\newcommand{\YS}{ {\texttt{Set}^{\cY}}}

\newcommand{\F}{ {\texttt{Free}}}
\newcommand{\BF}{ {\texttt{Free}^{\BB}}}
\newcommand{\RF}{ {\texttt{Free}^{\RR}}}
\newcommand{\GF}{ {\texttt{Free}^{\GG}}}
\newcommand{\CF}{ {\texttt{Free}^{\cC}}}
\newcommand{\XF}{ {\texttt{Free}^{\cX}}}

\newcommand{\CNX}{ {\texttt{Not}\cX^{\cC}}}
\newcommand{\XNY}{ {\texttt{Not}\cY^{\cX}}}
\newcommand{\YNC}{ {\texttt{Not}\cC^{\cY}}}

\newcommand{\CNY}{ {\texttt{Not}\cY^{\cC}}}
\newcommand{\XNC}{ {\texttt{Not}\cC^{\cX}}}
\newcommand{\YNX}{ {\texttt{Not}\cX^{\cY}}}

\newcommand{\CU}{U^{\cC}}
\newcommand{\RU}{ {U^{\RR}}}
\newcommand{\GU}{ {U^{\GG}}}
\newcommand{\BU}{ {U^{\BB}}}

\newcommand{\CC}{ {\mathcal{C}}}
\newcommand{\PP}{ {\mathcal{P}}}

\newcommand{\LL}{ {\mathcal{L}}}

\begin{document}

  \date{}

  \author{
    Marek Cygan
    \and
    Marcin Pilipczuk
    \and
    Micha\l{} Pilipczuk
    \and
    Jakub Onufry Wojtaszczyk\thanks{
    Dept.~of Mathematics, Computer Science and Mechanics,
    University of Warsaw, Poland, 
    \texttt{[cygan@,malcin@,michal.pilipczuk@students.,onufry@]mimuw.edu.pl}}}

  \title{The stubborn problem is stubborn no more \\
\large{    (a polynomial algorithm for 3--compatible colouring and the stubborn list partition problem)}}

\begin{titlepage}
  \def\thepage{}
  \thispagestyle{empty}
  \maketitle

\begin{abstract}
One of the driving problems in the CSP area is the Dichotomy Conjecture, formulated in 1993 by Feder and Vardi~[STOC'93], stating that for any fixed relational structure $\Gamma$ the Constraint Satisfaction Problem \cspgamma{} is either NP--complete or polynomial time solvable. 
A large amount of research has gone into checking various specific cases of this conjecture. One such variant which attracted a lot of attention in the recent years is the \listpartitionname{} problem. 
In 2004 Cameron et al.~[SODA'04] classified almost all \listpartitionname{} variants for matrices of size at most four. 
The only case which resisted the classification became known as the \stubname. In this paper we show a result which enables us to finish the classification --- thus solving a problem which resisted attacks for the last six years.

Our approach is based on a combinatorial problem known to be at least as hard as the \stubname{} --- the \threeccname{} problem.
In this problem we are given a complete graph with each edge
assigned one of $3$ possible colours and we want
to assign one of those $3$ colours to each vertex 
in such a way that no edge has the same
colour as both of its endpoints. 
The tractability of the \threeccname{}
problem has been open for several years and the best
known algorithm prior to this paper is due to Feder et al.~[SODA'05]
 --- a quasipolynomial algorithm with a $n^{O(\log n / \log \log n)}$ time complexity.
In this paper we present a polynomial--time algorithm for 
the \threeccname{} problem and
consequently we prove a dichotomy for the $k$-{\sc Compatible Colouring}
problem.


\end{abstract}
\end{titlepage}

\newpage

\newcommand{\cN}[1]{\ensuremath{N[#1]}}
\newcommand{\naglowek}[1]{\noindent {\bf{#1}} }

\section{Introduction}

In this paper we consider a variant of the graph colouring problem, namely the \kccname{} problem.
We are given a complete graph with each edge assigned one of $k$ possible colours and we want
to assign one of those $k$ colours to each vertex in such a way that no edge has the same
colour as both of its endpoints. Formally:

\defproblemu{\kccname{} (sometimes called {\sc Edge Free} $k$-{\sc{Colouring}})}
{A complete undirected graph $G = (V,E)$ and a function $\CC:E\rightarrow \{0,\ldots,k-1\}$}
{Does there exist a function $\phi: V \rightarrow \{0,\ldots,k-1\}$ such that for each edge $uv \in E$ either $\phi(u) \not = \CC(uv)$ or $\phi(v) \not = \CC(uv)$}

For $k=1$ this problem is meaningless, but for $k=2$ it can be interpreted as a split graphs recognition problem.
Indeed, if we consider a graph $G' = (V,\CC^{-1}\{1\})$ (i.e., we take only those edges from $e\in E$ for which $\CC(e) = 1$)
our task is equivalent to partitioning the graph $G'$ into a clique and an independent set.
Graphs that can be partitioned in this way are called {\em split graphs} and can be 
recognized in linear time~\cite{golumbic}.

It is known~\cite{hell-nesetril} that for $k \ge 4$ the \kccname{} problem becomes NP-complete.
However, for $k=3$ the problem of its tractability has been open for several years.
In this paper we show a polynomial--time algorithm for this case.

To compare, the classical colouring problem is NP-complete for $k\ge 3$ and polynomial time solvable
for $k \le 2$.
Until now it was not known whether \kccname{} admits such a dichotomy since the previously best algorithm (by Feder et al. 
from~2005~\cite{kral-soda05}) has a $n^{O(\log n / \log \log n)}$ time complexity
which is an improvement over the $n^{O(\log n)}$ time complexity of an algorithm by Feder and Hell~\cite{feder-hell-full-csp}.

\paragraph{Related work and motivation}

We briefly sketch the \cspname{} (\cspshort{}) definition 
in the notation proposed by Feder and Vardi~\cite{feder-vardi}.
For a fixed relational structure $\Gamma$ in the problem \cspgamma{} 
we are given a second relational structure $G$ and we are asked 
whether there exists a homomorphism of $G$ to $\Gamma$ (a mapping $f: V(G) \rightarrow V(\Gamma)$
which preserves all the relations).
Feder and Vardi~\cite{feder-vardi} in 1993 formulated the following conjecture
which remains open and motivates a lot of research in this area.

\begin{conjecture}[The Dichotomy Conjecture~\cite{feder-vardi}]
\label{conj:dichotomy}
For any fixed relational structure $\Gamma$ the problem \cspgamma{} is either NP-complete
or polynomial time solvable.
\end{conjecture}

Since then dozens of papers have been written proving this conjecture
in several special cases (for a survey see~\cite{hell-nesetril}).
In particular, Conjecture~\ref{conj:dichotomy} holds for every
relational structure of size two~\cite{shaefer} and three~\cite{bulatov}.

The \kccname{} problem is a variant of \cspfull{} problems introduced by Feder and Hell in~\cite{feder-hell-full-csp}, whereas the exact name \kccname{} to the best of our knowledge
comes from~\cite{there-and-back}. 
Intuitively, in \cspfull{} problems we restrict ourselves to structures $G$ in which every tuple of elements is restricted by some constraint.
A similar variant of \cspshort{} studied in the literature is called the \listpartitionname{} 
where $\Gamma$ is represented by an $r \times r$ symmetrical matrix $M$
with entries being subsets of $\{0,\ldots,q-1\}$ for some integer $q$. We are given a complete graph $G$ with vertices equipped with subsets of $\{0,\ldots,r-1\}$ and edges assigned values from $\{0,\ldots,q-1\}$. We ask whether there exists a function $\phi: V(G) \rightarrow \{0,\ldots,r-1\}$ such that for each $v\neq w\in V(G)$, $\phi(v)$ belongs to the set tied to the vertex $v$ and the value associated with the edge $vw$ belongs to the set in the $\phi(v)$-th row and $\phi(w)$-th column of $M$. A formal description can be found in~\cite{fhkm}. It is known that for fixed $q,r$ the \listpartitionname{} problem enjoys a quasi-dichotomy.

\begin{theorem}[Quasi-dichotomy Theorem~\cite{feder-manuscript}]
\label{thm:quasi}
For each pair of positive integers $r,q$ and for each symmetrical $r \times r$ matrix $M$
whose entries are subsets of the set $\{0,\ldots,q-1\}$ 
the \listpartitionname{} problem is either NP-complete or solvable in
quasipolynomial time.
\end{theorem}

The currently best bound for the quasipolynomial from Theorem~\ref{thm:quasi} due to 
Feder and Hell~\cite{feder-manuscript} is $n^{O(\log n)}$, where $n=|V(G)|$.
In order to check whether this quasi-dichotomy is a classical dichotomy
several special cases for small values of $q$ and $r$ were studied.
In particular, Cameron et al.~\cite{cameron} were able to classify almost all matrices with $r \le 4$ and $q=2$.
For all classified matrices either a polynomial time algorithm or
a NP-completeness proof was given.
Interestingly enough, the classified cases were equivalent to numerous 
classical graph problems such as: $3$-colourability, clique cutset, stable cutset,
skew partition and split graphs recognition.
To underline the significance of the \listpartitionname{}
we recall (as stated in~\cite{cameron}) that the resolution 
of the {\em Strong Perfect Graph Conjecture}
by Chudnovsky et al.~\cite{chudnovsky} relies in part on decompositions
that can be formulated as \listpartitionname{} instances.
The only two matrices that Cameron et al. could not classify 
are polynomially equivalent to the following problem which came to be called the \stubname{}.

\defproblemu{\stubname}{
	An undirected graph $G = (V,E)$ and a constraint function $\LL : V \ra \PP(\{1,2,3,4\})$
	}{
	Does there exists a colouring $\phi : V \ra \{1,2,3,4\}$, for which $\phi(v) \in \LL(v)$, $\phi^{-1}(4)$ is a clique, and for any edge $uw \in E$ the set $\phi^{-1}(\{u,w\})$ is different from $\{1\}$, $\{2\}$ and $\{1,3\}$?
	}

It is known that a polynomial algorithm for the \threeccname{} problem
implies a polynomial algorithm for the \stubname{} (as stated in~\cite{feder-hell-full-csp}).
Due to their role as the the last unresolved case in the classification of Cameron et al., the problems attracted
quite a lot of attention.
In particular, the polynomial status of \threeccname{} or \stubname{}
was mentioned as an open problem in numerous places 
including~\cite{open-garden,cameron,dantas,dagstuhl,feder-manuscript,feder-hell-full-csp,kral-soda05,stacho,tucker,there-and-back}.

\paragraph{Our results}

In this paper we present a polynomial time algorithm for the \threeccname{} problem and hence for the \stubname{},
resolving a long standing open problem in the \cspfull{} dichotomy project:

\begin{theorem} There exists a $O(|(V,\CC)|^{3.5})$ algorithm for the \threeccname{} problem, where $|(V,\CC)|=O(|V|^2)$ is the size of the instance. \end{theorem}

\begin{theorem} There exists a $O(|G|^7)$ algorithm for the \stubname, where $|G|=O(|V|+|E|)$ is the size of the instance. \end{theorem}

Our results prove the dichotomy for the \kccname{} problem.
Moreover, combining with results by Cameron at al.~\cite{cameron} 
we finish the matrix classification up to size $4 \times 4$ 
for the \listpartitionname{} problem proving that quasi-dichotomy
can be strengthened to the classical dichotomy and hence improve 
results of Feder et al.~\cite{fhkm}.

\begin{theorem} Let $M$ be a symmetrical $r\times r$ matrix whose entries are subsets of $\{0,1\}$. If $r\leq 4$ then for $M$ the \listpartitionname{} problem is either NP-complete or solvable in polynomial time.
\end{theorem}


In the literature one can also find a list version of the \threeccname{} problem, where each vertex $v$
is additionally equipped with a set $S_v \subseteq \{\RR,\GG,\BB\}$; and we demand that the colouring we 
construct satisfies additionally $\phi(v) \in S_v$. 
It is known that the list version of the \threeccname{} problem can be reduced to the original \threeccname{}
problem for instance using gadgets described in Appendix~\ref{sec:stubborn}.

The \threeccname{} problem came to our attention when posted by Marx 
in the open problems list from Dagstuhl Seminar 09511 on 
{\em Parameterized complexity and approximation algorithms}~\cite{dagstuhl}.
Marx suspected that Fixed Parameter Tractability tools and intuitions may be useful
either to design a polynomial time algorithm or a quasi-polynomial lower bound.
While the final version of the algorithm is elementary and uses no tools from the
parametrized complexity setting, our reasoning was heavily influenced by a 
technique called {\em iterative compression}, developed by Reed et al. \cite{reed:ic}. 



\paragraph{Outline of the paper}
In Section~\ref{sec:2col} we investigate the structure of
solutions for the $2$-\ccname{} problem
(i.e., finding a split graph structure).
In Section~\ref{sec:algorithm} we present our algorithm
where Section~\ref{sec:correctness} is devoted to its correctness
and Section~\ref{sec:time} to its time complexity.
The correctness of our algorithm is not hard,
hence an advanced reader may skip this section.
However, the proof of the time complexity of our algorithm is not trivial
and relies on interesting combinatorial facts included in Lemma \ref{lem:splitCF}.

We were unable to find a reduction from the \stubname{} to the \threeccname{} problem 
in literature. Hence for the sake of completeness,
we present our own reduction in Appendix~\ref{sec:stubborn}.

\paragraph{Notation}
We assume that we are given an input to the \threeccname{} problem:
an undirected complete graph $G=(V,E)$ with a colouring of edges $\CC:E \to \{\RR,\GG,\BB\}$
(we denote the colours by $\RR$, $\GG$ and $\BB$). For a subset of vertices $X \subseteq V$
by $G[X]$ we denote the subgraph induced by $X$.
For a subset of edges $E' \subseteq E$ by $V(E')$ we denote the set of all endpoints of edges in $E'$.
Similarly, for a subset of vertices $V' \subseteq V$ by $E(V')$ we denote the set of edges
with both endpoints in the set $V'$.

\section{Colouring with two colours --- preliminaries}\label{sec:2col}
We first consider the structure of \twoccname. 
Let $W$ be such a set of vertices that $\CC$ restricted to $E(W)$ has only two values, say $\RR$ and $\BB$. 
We look for all feasible colourings $\phi: W \ra \{\RR,\BB\}$. 

\begin{definition} 
We say a vertex $v \in W$ is {\em interesting} if there exist two feasible colourings $\phi_1$, $\phi_2$ of $W$ into $\RR$ and $\BB$ such that $\phi_1(v) = \RR$ and $\phi_2(v) = \BB$. 
Otherwise a vertex is {\em boring}. 
\end{definition}

In particular, if there is no feasible colouring of $W$, all vertices of $W$ are boring.

\begin{lemma}\label{lem:boring} 
Let $u,v,w$ be three such vertices in $W$ that the $\CC(uv) = \CC(vw) \neq \CC(uw)$. 
Then $v$ is boring, as it does not admit a feasible colouring with $\phi(v) = \CC(vw)$. 
\end{lemma}

\begin{proof} 
Assume without loss of generality that $\CC(uv) = \CC(vw) = \RR$ and $\CC(uw) = \BB$. 
Assume there is a feasible colouring $\phi$ of $\{u,v,w\}$ in which $\phi(v) = \RR$. 
Then we would have to have $\phi(u) = \phi(w) = \BB$ (as $\CC(uv) = \CC(vw) = \RR$), but this contradicts $\CC(uw) = \BB$. 
As any feasible colouring of $W$ restricted to $\{u,v,w\}$ is a feasible colouring of $\{u,v,w\}$, $v$ cannot be interesting.
\end{proof}

\begin{lemma}\label{lem:inter}\label{lem:boringalg} 
Let $I \subseteq W$ be the set of interesting vertices in $W$. 
Then $\CC$ restricted to $E(I)$ has only one value (that is all the edges in $E(I)$ are of a single colour).

Moreover, there exists an algorithm which either finds a boring vertex and the colour it cannot have, or returns NO if all vertices are interesting. The algorithm works in $O(|W|^2)$ time.
\end{lemma}

\begin{proof} 
If all the edges of $E(W)$ are of the same colour (without losing generality $\RR$), 
every vertex $v\in W$ is interesting, 
as when one sets $\phi(w)=\BB$ for $w\neq v$, 
then any value of $\phi(v)$ makes $\phi$ a feasible colouring. 
Therefore, in this case the answer of the algorithm is ,,NO''.
This check can be performed in $O(|W|^2)$ time.

Now assume we found two edges $u_1v_1$ and $u_2v_2$ of different colours.
If these edges share an endpoint, e.g. $v_1=v_2$, 
then there is a multicoloured triangle $(u_1,u_2,v_1)$. 
On the other hand if all the endpoints are different, then the edge $u_1u_2$ has a different colour from one of the edges $u_1v_1$, $u_2v_2$. 
Therefore, one of the triples $(u_1,u_2,v_1)$ or $(u_1,u_2,v_2)$ forms a multicoloured triangle. 
In each case the multicoloured triangle gives us a boring vertex with its inadmissible colour as in Lemma \ref{lem:boring} in constant time.
\end{proof}

\section{The algorithm}\label{sec:algorithm}

\subsection{Outline of the algorithm}

Let $v_1,v_2,\ldots,v_n$ to be an arbitrary order on $V$. 
Suppose we have an instance $(V,\CC)$ of the \threeccname{} problem. 
Let $V_i = \{v_1,v_2,\ldots,v_i\}$, and let $\CC_i$ be the restriction of $\CC$ to edges in $E(V_i)$. 
Notice that if $\phi$ is a solution for $(V_i,\CC_i)$, then $\phi$ restricted to $V_j$ is a solution to $(V_j,\CC_j)$ for any $j < i$. 
Thus, in particular, if there is a positive answer to $(V,\CC)$, then there is a positive answer to any $(V_i,\CC_i)$.

We proceed by building a solution for each $(V_i,\CC_i)$. 
Obviously we may start the induction with an empty set $V_0$ and empty function $\CC_0$.
If for some $i$ we show there is no solution, we return NO as an answer to the original $(V,\CC)$ instance. 
Moreover, when building the solution to $(V_i,\CC_i)$ we assume we are given some solution to $(V_{i-1},\CC_{i-1})$. 
Thus, we can focus on a situation in which we solve an instance $(V,\CC)$ and
we already have a feasible colouring $\phi_0$ for $(V \setminus \{v_0\}, \CC)$ for one fixed vertex $v_0$.
We use this feasible colouring $\phi_0$ to deeply exploit the colouring of the graph $G[V \setminus \{v_0\}]$
which is a crucial part in designing our algorithm.
This type of reasoning is one of the key parts of the aforementioned {\em iterative compression} technique used in the Fixed Parameter Tractability community.

In each step of the algorithm we have a division of $V$ into eighteen sets, six corresponding to each of the three colours. 
The algorithm is a branching algorithm --- we perform operations which either simply move the vertices around, or branch out into several instances. 
Then we resolve each branch recursively, and if we find a feasible colouring in any of them, we return this colouring, while if all the branches return NO, we return NO.
We follow a naming convention in which if $\cC$ is one of the colours in $\{\RR, \GG, \BB\}$, then $\cX$ and $\cY$ are the other two.

Consider any colour $\cC \in \{\RR,\GG,\BB\}$. 
The sets corresponding to this colour are $\CF$, $\CTD$, $\CS$, $\CTS$, $\CNX$ and $\CNY$. 
The intuitive meanings of these sets are as follows:
\begin{itemize}
\item $\CF$ --- the ``free'' vertices of colour $\cC$ --- those, which were of colour $\cC$ in $\phi_0$ and our algorithm has not yet gained any information about them;
\item $\CTD$ --- the ``to do'' vertices of colour $\cC$ --- those, which were of colour $\cC$ in $\phi_0$, but our algorithm already learned they will not be of colour $\cC$ in the new colouring;
\item $\CS$ --- the ``set'' vertices of colour $\cC$ --- those which our algorithm has already determined to be of colour $\cC$;
\item $\CTS$ --- the ``to set'' vertices of colour $\cC$ --- those which are determined to be of colour $\cC$ in the new colouring, but we have to update the current division of $V$  before we put them into $\CS$;
\item $\CNX$ and $\CNY$ --- the ``not $\cX$'' and ``not $\cY$'' vertices of colour $\cC$ --- those which were of colour $\cC$ in $\phi_0$, and we already know they will not be of colour $\cX$ (or $\cY$, respectively) in the new colouring.
\end{itemize}

This information can be represented by associating with each vertex the colour assigned to it by
$\phi_0$ and the subset $S(v) \subseteq \{\RR,\GG,\BB\}$ of colours which are still admissible as 
values of $\phi(v)$.
Such an approach would certainly streamline any implementation of the algorithm, but
we think that naming each set separately helps underline the role each particular set plays
in the analysis --- thus the choice of this method of presentation.

To start the algorithm we put $\phi_0^{-1}(\RR)$ into $\RF$, $\phi_0^{-1}(\BB)$ into $\BF$ and $\phi_0^{-1}(\GG)$ into $\GF$. 
There are three possible colours we can give to $v_0$, thus we branch out into three cases, putting $v_0$ into $\BTS$, $\RTS$ or $\GTS$.

The algorithm uses two subprocedures --- shifting a vertex (from $\CTS$ to $\CS$) and resolving a set $\CTD$. 
As long as any of the sets $\CTS$ is non--empty, we shift vertices from this set. 
If all sets $\CTS$ are empty, but there is a non--empty set $\CTD$, we resolve the set $\CTD$. 
If all the sets $\CTS$ and $\CTD$ are empty, we claim that setting $\phi(v) = \cC$ for $v \in \CS \cup \CF \cup \CNX \cup \CNY$ is a feasible solution and return it.

\subsection{Shifting a vertex}
The meaning of this step is that we have a vertex $v$ for which we have just determined that $\phi(v) = \cC$. 
This gives us some information about the vertices $w$ with $\CC(vw) = \cC$, which we represent by moving vertices between appropriate sets.
After including the gained information in our structure we can safely move $v$ into $\CS$.

Let $v \in \CTS$. 
The procedure of shifting a vertex works as follows: we move $v$ from $\CTS$ to $\CS$, and then consider all $w \in V$ such that $\CC(vw) = \cC$. 
For each such vertex $w$ we perform the appropriate action (in parentheses we give the intuitive meanings of the actions). As before, $\cX$ denotes any colour different than $\cC$ and $\cY$ denotes
the third colour different than $\cC$ and $\cX$.
\begin{itemize}
\item If $w \in \CS$ return NO from this branch (we have two vertices for which $\phi(v) = \phi(w) = \cC$ connected with an $\cC$--edge);
\item If $w \in \CF$ move $w$ to $\CTD$ ($w$ cannot be of colour $\cC$);
\item If $w \in \CNX$ move $w$ to $\YTS$, where $\cY$ is the third colour, that is $\{\cC,\cX,\cY\} = \{\RR,\GG,\BB\}$ (it is not of colour $\cC$ nor $\cX$, thus it is of colour $\cY$);
\item If $w \in \XTD$ move $w$ to $\YTS$, where $\cY$ is as above (again, $w$ is neither of colour $\cX$ nor $\cC$, so it is of colour $\cY$);
\item If $w \in \XF$ move $w$ to $\XNC$ ($w$ cannot be of colour $\cC$);
\item If $w \in \XNY$ move $w$ to $\XTS$ ($w$ cannot be of colour $\cY$ nor $\cC$);
\item If $w \in \CTD$, $w \in \CTS$, $w \in \XNC$, $w \in \XTS$ or $w \in \XS$, do nothing.
\end{itemize}

\subsection{Resolving a set}
Consider a non--empty set $\CTD$. 
The meaning of this step is that we have a set of vertices that were of colour $\cC$ in $\phi_0$, but we see they cannot be of colour $\cC$ in $\phi$.
Thus, there are no $\cC$--edges in $E(\CTD)$, and we have to colour $\CTD$ into the two remaining colours.
If there are any boring vertices in $\CTD$, we know how to colour them, so we move them to appropriate $\TS$ sets and go back to shifting vertices.
If all vertices in $\CTD$ are interesting, we find all possible colourings of $\CTD$ and branch out.

We prove formally that $E(\CTD)$ contains no edges of colour $\cC$ in Section \ref{sec:correctness}.
Apply the algorithm from Lemma \ref{lem:boringalg} to $\CTD$.
If we find any boring vertex $v$ which does not admit colour $\cX$, we move it to $\YTS$ and finish the resolving step.
If all vertices in $\CTD$ are interesting, we branch out into $|\CTD| + 1$ cases.
We know that all the edges of $E(\CTD)$ are of one colour by Lemma \ref{lem:inter}.
We check a single edge to find out which colour it is, without loss of generality assume it is $\cX$. If $|\CTD|=1$ and such an edge does not exist, it does not matter which colour 
different than $\cC$ we choose.
In one branch we move the whole set $\CTD$ to $\YTS$.
In the other $|\CTD|$ branches we choose one vertex $v \in \CTD$, a different one in each branch, and move this vertex to $\XTS$ and all the other vertices to $\YTS$. Note that 
these branches correspond to all feasible colourings of $\CTD$ using colours different than $\cC$.
Then we solve each branch recursively, if any of them returns a feasible colouring, we return it, while if all of them return NO, we return NO.

\section{Correctness of the algorithm}\label{sec:correctness}

We formally prove the correctness of the algorithm given in Section \ref{sec:algorithm}. A reader accustomed to such algorithms may probably only glance over this section and fill in the necessary details by him- or herself.

Formally, we do not yet know that the algorithm always terminates. In order to clarify the proof, we now assume that this indeed holds. 
In Section \ref{sec:time} we justify this assumption by showing even polynomial bounds on the algorithm's working time.

\begin{definition} We say a division of $V$ into the eighteen sets satisfies {\em proper invariants} if
\begin{enumerate}
\item For each colour $\cC$ and for any $e \in E(\CS \cup \CF \cup \CNX \cup \CNY)$ we have $\CC(e) \neq \cC$;
\item For each colour $\cC$ and for any $e \in E(\CTD \cup \CF \cup \CNX \cup \CNY)$ we have $\CC(e) \neq \cC$;
\end{enumerate}
\end{definition}

\begin{definition}
A colouring $\phi: V \ra \{\RR,\GG,\BB\}$ is said to be {\em proper} with respect to a division of $V$ into the eighteen sets if for every colour $\cC$ it satisfies
\begin{itemize}
\item $\phi(v) \neq \cC$ for $v \in \XNC, \YNC, \CTD$;
\item $\phi(v) = \cC$ for $v \in \CTS$, $v \in \CS$;
\end{itemize}
\end{definition}

We prove that the division at each step of our algorithm satisfies proper invariants.
Moreover, we prove that if there exists a proper solution $\phi$, then our algorithm does not return NO.

\subsection{Proper invariants}
Note that as $\phi_0$ was a feasible colouring for $V \setminus \{v_0\}$, the proper invariants are satisfied at the start of the algorithm.

We have to check that the operations of shifting a vertex and resolving a set do not spoil proper invariants. 

Firstly, we consider shifting a vertex.
Assume we shift a vertex $v$ from $\CTS$ to $\CS$.
Begin by considering the moves of vertices $w$ with $\CC(vw) = \cC$. 
The moves $\CNX \ra \YTS$, $\XTD \ra \YTS$ and $\XNY \ra \XTS$ cannot spoil proper invariants since the sets $\CTS$ are not involved in the invariants. 
Returning NO obviously does not spoil proper invariants. 
The move $\CF \ra \CTD$ decreases the number of constraints in the invariants, and $\XF \ra \XNC$ does not change the invariants.

As far as the move of $v$ from $\CTS$ to $\CS$ is concerned, if there were any vertices $w \in \CS \cup \CNX \cup \CNY \cup \CF$ such that $\CC(vw)=\cC$, the shifting algorithm removes them from the set (or returns NO for $w \in \CS$).

Thus after shifting a single vertex proper invariants still hold.

Resolving a set involves only moving vertices to the $\TS$ sets, which are not constrained in the invariants, so it does not spoil the invariants as well.

\subsection{Existence of a solution}
Assume that at a given stage of the algorithm there is a proper colouring $\phi$, which is a feasible solution to $(V,\CC)$. 
We prove that after performing a single step $\phi$ is still proper in at least one branch.

First consider shifting a vertex $v$ from $\CTS$ to $\CS$. 
As $v$ was in $\CTS$ and $\phi$ is proper, $\phi(v) = \cC$.
Thus after moving $v$ from $\CTS$ to $\CS$ the solution $\phi$ is still proper.
Consider any vertex $w$ with $\CC(vw) = \cC$.
Then $\phi(w) \neq \cC$. 
If $w \in \CS$ we have a contradiction as $\phi$ being a proper solution implies $\phi(w) = \cC$.
If $w$ is moved to $\CTD$ or $\XNC$ (from $\CF$ or $\XF$, respectively), $\phi$ is still a proper solution, for the only new constraint is that $\phi(w) \neq \cC$, which we know to be satisfied.
If $w$ was in $\CNX$, $\XTD$ or $\YNX$, then $\phi(w) \neq \cX$ as $\phi$ was proper.
As we additionally know that $\phi(w) \neq \cC$, this implies $\phi(w) = \cY$, thus after moving $w$ to $\YTS$ the solution $\phi$ remains proper.
Thus $\phi$ is still proper after shifting a vertex.

Now consider resolving a set $\CTD$. 
As $\phi$ is proper, $\phi(v) \neq \cC$ for any $v \in \CTD$. 
On the other hand, the proper invariants guarantee that $\CC(e) \neq \cC$ for $e \in E(\CTD)$. 
Thus the application of Lemma \ref{lem:boringalg} is justified. 
If there exists a boring $v \in \CTD$, which --- according to the algorithm from Lemma \ref{lem:boringalg} --- cannot have $\phi(v) = \cX$ for any feasible colouring, we have $\phi(v) = \cY$.
Thus after moving $v$ to $\YTS$ the solution $\phi$ remains proper.

If all vertices are interesting, then by Lemma \ref{lem:inter} all the edges in $E(\CTD)$ are of a single colour, say $\cX$, thus at most one vertex $v \in \CTD$ satisfies $\phi(v) = \cX$.
If there exists such a vertex, $\phi$ is a proper colouring for the branch in which we move $v$ to $\XTS$ and all the other vertices from $\CTD$ to $\YTS$.
If no such vertex exists, $\phi$ is a proper colouring for the branch in which we move all vertices to $\YTS$.

Now assume that there exists any solution $\phi$ for the original problem $(V,\CC)$. 
Let $\cC = \phi(v_0)$. 
Then $\phi$ is proper in the starting branch in which we set $v_0 \in \CTS$ --- we have $\phi(v_0) = \cC$, and all the other vertices are in the sets $\F$, so we assume nothing about them. 

So, finally --- if there exists a solution for the original problem, our algorithm returns a solution.
On the other hand, if our algorithm returns a solution, sets $\CTD$ and $\CTS$ are empty and the first proper invariant guarantees that it is a feasible solution to the original problem.

This allows us to formulate the following theorem:
\begin{theorem} 
Consider an instance $(V,\CC)$ of compatible colouring, and assume we are given a feasible colouring $\phi$ for $(V \setminus \{v\}, \CC)$. 
Then if there exists any feasible colouring for $(V,\CC)$, the algorithm described in Section \ref{sec:algorithm} returns a colouring, and conversely any colouring returned by the algorithm is a feasible one for $(V,\CC)$.
\end{theorem}

\section{Time complexity bounds}\label{sec:time}

Let us denote $|V|$ by $n$.
Consider a tree of recursion for our algorithm.
We actually consider three recursion trees, one for each possible choice of the set $\CTD$ to put $v_0$ into.
\begin{definition}
By a {\em state} $S$ of the algorithm we mean a division of the set $V$ into the eighteen sets postulated by the algorithm. We denote these 18 sets by $\CF(S)$, $\CTD(S)$, and so on, omitting the argument when it is clear what state we are considering.

By an {\em inner node} of the recursion tree we mean the state of the algorithm at a moment just before branching out when resolving a set $\CTD$ with no boring vertices.

By a {\em leaf node} of the recursion tree we mean the state of the algorithm when it terminates a branch 
--- either answering NO due to a failed shift operation or returning a solution due to the sets $\CTD$ and $\CTS$ all being empty.
For the sake of analysis it is better to assume that when answering NO we first shift all vertices out the $\CTS$ sets, disregarding conflicts, and move the vertices required by the shift.
Thus we answer NO in the state when all sets $\CTS$ are empty.

By the {\em descendants} of an inner node $N$ we mean nodes that occur in any of the branches of resolving $\CTD$ in $N$.
This obviously gives rise to a tree structure in each of the three recursion trees, so we use the standard terms ``child'', ``father'', ``root'' and so on.
\end{definition}
Each branching out takes $O(n^2)$ time for the application of Lemma \ref{lem:boringalg} (not counting the time needed to solve the branches), and in total $O(n^2)$ time to prepare the branches.
Between an inner node and its child a number of operations are performed, each being either shifting a single vertex (which takes $O(n)$ time) or resolving a set containing boring vertices (which takes $O(n^2)$ time).

\subsection{Length of branches}

\begin{definition} 
The {\em potential} of a given state $S$ of the algorithm is equal to
$$\sum_{\cC \in \{\RR,\GG,\BB\}} 3|\CF(S)| + 2|\CTD(S)| + 2|\CNX(S)| + 2|\CNY(S)| + |\CTS(S)|.$$
\end{definition}

\begin{lemma} 
Shifting a single vertex and resolving a set containing a boring vertex decreases the potential. 
\end{lemma}

\begin{proof} 
The move $\CTS \ra \CS$, which happens every time we shift a vertex, decreases the potential by $1$. 
The same holds for the move $\CTD \ra \XTS$, which happens every time we resolve a set with a boring vertex.
All the other moves associated shifting a vertex ($\CF \ra \CTD$, $\CNX \ra \YTS$, $\XTD \ra \YTS$, $\XF \ra \XNC$ and $\XNY \ra \XTS$) do not increase the potential.
\end{proof}

\begin{lemma}\label{lem:branch-non-increase}
When we branch out while resolving a set without boring vertices, the potential in each of the branches is smaller than the potential in the original state. 
\end{lemma}

\begin{proof}
We resolve only non--empty sets.
We move all vertices from $\CTD$ to $\XTS$ or $\YTS$, each such move decreases the potential by one.
\end{proof}

The starting potential is $O(n)$, and decreases with each operation. Thus we have the following corollary:

\begin{corollary} \label{cor:pathlen}
We perform $O(n)$ operations (i.e., shifts, resolves of boring vertices or branches)
on each path from a starting node to any leaf of the recursion tree.
\end{corollary}

\subsection{Number of leaves}

We begin by formulating the lemma which is crucial to estimating the number of leaves:

\begin{lemma} \label{lem:splitCF}
Consider any inner node $S$ of the recursion tree formed immediately before resolving a set $\CTD$ containing no boring vertices.
Let $S_0, \ldots, S_k$ be the children of $S$.
Let $\CU_i = \CF(S_i) \cup \CTD(S_i)$.
Then the sets $\CU_i$ are disjoint subsets of the set $\CF(S)$.
\end{lemma}

\begin{proof}
The sets $\CU_i$ are subsets of $\CF(S)$ by the definition of resolving a set.
By application of Lemma \ref{lem:inter} we conclude that $E(\CTD(S))$ contains edges of a single colour (different than $\cC$ due to the proper invariants), say $\cX$ (if $\CTD(S)$ consists of a single vertex, take as $\cX$ any colour different than $\cC$).
Denote the vertices of $\CTD$ by $v_1, v_2, \ldots, v_k$. Without losing generality assume that $S_0$ corresponds to the branch where the whole $\CTD$ is moved to $\YTS$, while $S_i$ for $i\geq 1$ corresponds to the branch where the vertex $v_i$ is the only one moved to $\XTS$.
Let $A_i$ be the set of those vertices $w$ in $\CF(S)$ for which $\CC(v_iw) = \cY$.
By the second proper invariant we know that for $w \in \CF(S) \setminus A_i$ we have $\CC(v_iw) = \cX$ --- there are no $\cC$--edges in $\CF \cup \CTD$.

Consider the branch in which we move the whole set $\CTD$ to $\YTS$. 
When we shift any vertex $v_i$ to $\YS$, every vertex $w \in A_i$ that was still left in $\CF$ is moved to $\CNY$.
Similarly, every vertex $w\in A_i$ now contained in $\CTD$ is moved to $\XTS$.
Also, neither a shift nor resolving a boring vertex moves any vertex into $\CF \cup \CTD$.
Thus after all the $k$ shifts of vertices that were in $\CTD$ before branching, $\CU_0$ is disjoint from $A_1 \cup A_2 \cup \ldots \cup A_k$.
Similarly, in the branch where $v_i$ is moved to $\XTS$ and the other $v_j$s are moved to $\YTS$, after all the shifts $\CU_i$ is disjoint from $A_1 \cup A_2 \cup \ldots \cup A_{i-1} \cup (\CF \setminus A_i) \cup A_{i+1} \cup \ldots \cup A_k$.

Consider any two branches and the associated sets $\CU_i$. 
Assume the first of these branches moved the $j$th vertex to $\XTS$ (at least one of them had to move some vertex to $\XTS$). 
Then $\CU_j$ for the first branch is contained in $A_j$, while $\CU_i$ for the second is disjoint from $A_j$.
This proves the thesis.
\end{proof}

We aim to prove that each recursion tree has $O(n^3)$ leaf nodes.
Consider the following definition:
\begin{definition} 
Let $\CU(S) = \CF(S) \cup \CTD(S)$, as above.
The {\em mass} of a given state $S$ of the algorithm is equal to
$$m(S) = (|\RU(S)| + 1)(|\GU(S)| + 1)(|\BU(S)| + 1).$$
\end{definition}

The mass of the root of the recursion tree is obviously $O(n^3)$, while the mass of each leaf is at least $1$.
As previously, shifting a vertex and resolving a boring vertex do not increase the mass of a state, as they cannot increase the sizes of sets $\RU,\GU,\BU$.
We will prove that for any node of the tree the mass of the node is not smaller than the sum of masses of its sons.
Clearly this leads to the conclusion that the mass of the root node is greater or equal to the sum of masses of all the leaves, which, along with the bounds for the masses of the root and the leaves, shows that there are at most $O(n^3)$ leaves.
Therefore, all we need is the following lemma:
\begin{lemma}\label{lem:splitmass}
Consider any node $S$ of the recursion tree formed immediately before resolving a set $\CTD$ containing no boring vertices. Let $S_0, S_1,\ldots, S_k$ be the children of $S$.
Then $$m(S) \geq \sum_{i=0}^k m(S_i).$$\end{lemma}

\begin{proof}
Without loss of generality, assume that we are resolving the set $\RTD$ in $S$.
Neither resolving a set $\RTD$, shifting a vertex nor resolving a boring vertex can increase the size of sets $\GU,\BU$, so for all $i=0,1,\ldots,k$ we have that $|\GU(S_i)| \leq |\GU(S)|$ and $|\BU(S_i)| \leq |\BU(S)|$.
Moreover, the number of branches (that is, the number of sons of $S$) is equal exactly to $k+1=|\RTD(S)| + 1$ --- one branch for every vertex in $\RTD$ to be assigned the ``other'' colour, and one branch for all the vertices having the same colour.
Thus, application of Lemma \ref{lem:splitCF} immediately yields:
$$\sum_{i=0}^k (|\RU(S_i)| + 1) = |\RTD(S)| + 1 + \sum_{i=0}^k |\RU(S_i)| \leq |\RTD(S)| + 1 + |\RF(S)| = |\RU(S)| + 1.$$
Multiplying this inequality by $(|\GU(S)| + 1)(|\BU(S)| + 1)$ we get
$$m(S) = (|\RU(S)| + 1)(|\GU(S)| + 1)(|\BU(S)| + 1) \geq \sum_{i=0}^k (|\RU(S_i)| + 1) (|\GU(S)| + 1)(|\BU(S)| + 1) \geq \sum_{i=0}^k m(S_i).$$
\end{proof}

As there are $O(n)$ operations on the path to each leaf, and each operation takes $O(n^2)$ time, we have the following corollary:

\begin{corollary} 
The total run--time of the algorithm described in Section \ref{sec:algorithm} is $O(n^{6})$
for each new vertex $v_0$. The whole algorithm runs in $O(n^{7})$ time.
\end{corollary}

\section*{Acknowledgements}
We would like to thank Daniel Marx for showing us this problem,
and for a number of suggestions that helped make this paper significantly better,
especially regarding Lemma~\ref{lem:splitmass}.

\newpage

\bibliographystyle{plain}
\bibliography{compatible-colouring}

\newpage

\appendix

\section{Stubborn problem reduction}\label{sec:stubborn}

Recall that the {\em stubborn problem} \cite{cameron} can be defined as follows:

\defproblemu{\stubname}{
	An undirected graph $G = (V,E)$ and a constraint function $\LL : V \ra \PP(\{1,2,3,4\})$
	}{
	Does there exists a colouring $\phi : V \ra \{1,2,3,4\}$, for which $\phi(v) \in \LL(v)$, $\phi^{-1}(4)$ is a clique, and for any edge $uw \in E$ the set $\phi^{-1}(\{u,w\})$ is different from $\{1\}$, $\{2\}$ and $\{1,3\}$.
	}

We show that this problem can be reduced to the \threeccname{} problem. 

\subsection{Gadgets}
We begin by showing two gadgets which can be implemented in \threeccname{}. Consider any \threeccname{} instance $(V,\CC)$. 
\begin{definition}
By adding a {\em type one $\cC$--gadget} to $G$ we mean adding $4$ vertices $v_1,v_2,v_3,v_4$ with
$\CC(v_1v_2) = \CC(v_3v_4) = \cX$ and 
$\CC(v_1v_3) = \CC(v_1v_4) = \CC(v_2v_3) = \CC(v_2v_4) = \cY$, 
and for any $v$ outside the gadget we have $\CC(v_1v) = \CC(v_2v) = \CC(v_3v) = \CC(v_4v)$.
The exact restraints can be defined arbitrarily.
\end{definition}

\begin{lemma}\label{lem:typ1}
Consider an instance $(V,\CC)$ of \threeccname{} and a set $S \subset V$. 
Let $(V',\CC')$ be $V$ after adding a type one $\cC$--gadget. 
We put $\CC(v_iv) = \cC$ for $v \in S$ and $\CC(v_iv) = \cX$ for $v \in V \setminus S$. 
Then $(V',\CC')$ has a solution iff $(V,\CC)$ has a solution $\phi$ with $\phi^{-1}(\cC) \cap S = \emptyset$.
\end{lemma}

\begin{proof}
If $(V,\CC)$ has a solution as above, we put $\phi' = \phi$ on $V$, $\phi'(v_i) = \cC$.
This is trivially a solution to $(V',\CC')$.

On the other hand, direct check shows that any feasible colouring of $\{v_1,v_2,v_3,v_4\}$ has to have at least one vertex of colour $\cC$. 
Thus any feasible colouring of $V'$ restricted to $V$ is a restricted colouring satisfying the conditions above.
\end{proof}

As a corollary we deduce that adding type one gadgets enable us to implement constraint lists --- in addition to the standard \threeccname{} structure we can demand that an arbitrary set of vertices {\em is not} of colour $\RR$ (or $\GG$ or $\BB$) by adding a type one $\RR$--gadget and connecting it to the set by edges of colour $\RR$.
Further on we assume we added to the graph type one gadgets of all three colours.

\begin{definition}
Let $u,w \in G$. 
By adding a {\em type two $\cC$--gadget} to $uw$ we mean adding two vertices $v_0,v_1$ with 
$\CC(uv_0) = \CC(wv_1) = \cC$,
$\CC(uv_1) = \CC(wv_0) = \cX$,
$\CC(v_0v_1) = \cY$.
Moreover, we assume both $v_0$ and $v_1$ are connected by $\cX$ edges to a type one $\cX$--gadget.
All the other edges connecting $v_0$ and $v_1$ to the graph are also $\cX$ edges.
\end{definition}

\begin{lemma}\label{lem:typ2}
Let $(V,\CC)$ be an instance of \threeccname{} and let $(V',\CC')$ be the same instance after adding a type two $\cC$--gadget to the edge $uw$.
Then $(V',\CC')$ has a feasible colouring iff $(V,\CC)$ has a feasible colouring in which at least one endpoint of $uw$ is not of colour $\cC$.
\end{lemma}

\begin{proof}
Note that in any feasible colouring of $(V',\CC')$ neither $v_0$ nor $v_1$ can be coloured $\cX$ due to the type one $\cX$--gadget. 
Moreover, at least one of them is not coloured $\cY$ due to the $\cY$--edge connecting them.
Thus at least one of them has to be coloured $\cC$, and --- due to the $\cC$--edges $v_0u$ and $v_1w$ --- at least one of $uw$ has to be of a colour different than $\cC$.
Thus the restriction of a feasible colouring on $(V',\CC')$ to $(V,\CC)$ is a colouring as above.

On the other hand, any colouring of $(V,\CC)$ as above can be extended to a proper colouring of $(V',\CC')$ by putting $\phi'(v_0) = \cC$ if $\phi'(u) \neq \cC$ and $\phi'(v_0) = \cY$ if $\phi'(u) = \cC$, and the same for $v_1$ and $w$.
\end{proof}

This gadget allows us to add additional edge constraints to the graph (as if we were able to draw multiple edges).

\subsection{Reduction}
Consider any instance $((V,E),\LL)$ of the \stubname{} problem. 
We construct an equivalent instance $(V',\CC')$ of \threeccname{} as follows:

\begin{itemize}
\item If $uw \in E$, put $\CC'(uw) = \RR$;
\item If $uw \notin E$, put $\CC'(uw) = \BB$;
\item Add type one gadgets of all three colours to $V$;
\item If $2 \notin \LL(v)$, connect $v$ to the $\RR$--gadget by $\RR$--edges;
\item If $4 \notin \LL(v)$, connect $v$ to the $\BB$--gadget by $\BB$--edges;
\item If $\{1,3\} \cap \LL(v) = \emptyset$, connect $v$ to the $\GG$--gadget by $\GG$--edges;
\item If $uw \in E$ and $3 \notin \LL(u)\cap \LL(w)$, add a type two $\GG$--gadget to $uw$.
\item All the edges connecting a type one $\cC$--gadget to the rest of the graph not defined above are $\cX$--edges.
\end{itemize}

We set out to prove the following theorem: 

\begin{theorem} There exists a feasible solution to the instance $((V,E),\LL)$ of \stubname{} iff there exists a feasible solution to the instance $(V',\CC')$ of \threeccname. \end{theorem}

\begin{proof}
If we have a solution $\phi$ to $((V,E),\LL)$, consider the following $\phi'$ for $v\in V$: if $\phi(v) = 2$, we put $\phi'(v) = \RR$, if $\phi(v) = 4$, we put $\phi'(v) = \BB$ and if $\phi(v) \in \{1,3\}$ we put $\phi'(v) = \GG$.
This is a feasible solution to $(V',\CC')$ without the added gadgets --- as there are no $\GG$ edges between vertices from $V$, $\phi^{-1}(4)$ is a clique so there are no $\BB$--edges connecting two $\BB$--vertices, and $\phi^{-1}(2)$ is an independent set, so there are no $\RR$--edges connecting two $\RR$--vertices.
Moreover, note that as the list constraints for $((V,E),\LL)$ were satisfied, the type one gadget constraints are satisfied in $(V',\CC')$.
Finally, for any two $\GG$--vertices we either have $uw \notin E$ or both of them were given $\phi(u)=\phi(w)=3$. Thus list constraints for $u$ and $w$ allowed value $3$, so there was no $\GG$--gadget on $uw$. So the type two gadget constraints are satisfied as well. Therefore, using Lemmata \ref{lem:typ1} and \ref{lem:typ2} one can extend $\phi'$ on the whole $V'$ obtaining a solution to $(V',\CC')$.

In the other direction, considering a feasible solution $\phi'$ to $(V',\CC')$ we can obtain a feasible solution to $((V,E),\LL)$ by putting $\phi(v) = 2$ for $\phi'(v) = \RR$, $\phi(v) = 4$ for $\phi'(v) = \BB$, $\phi(v) = 3$ if $\phi'(v) = \GG$ and $3 \in \LL(v)$ and $\phi(v) = 1$ if $\phi'(v) = \GG$ and $3 \notin \LL(v)$.
\end{proof}

\end{document}